\documentclass{amsart}
\usepackage{amsmath}
\usepackage{amsthm}
\usepackage{amssymb}
\usepackage{amsfonts,times}
\usepackage{hyperref}

\theoremstyle{plain}
\newtheorem{theorem}{Theorem}[section]

\newtheorem{lemma}[theorem]{Lemma}

\newtheorem{proposition}[theorem]{Proposition}

\theoremstyle{remark}
\newtheorem{remark}[theorem]{Remark}

\def\R{{\mathbb R}}
\def\N{{\mathbb N}}
\def\O{\mathcal O}

\def\({\left(}
\def\){\right)}
\def\<{\left\langle}
\def\>{\right\rangle}
\def\le{\leqslant}
\def\ge{\geqslant}

\def\Eq#1#2{\mathop{\sim}\limits_{#1\rightarrow#2}}
\def\Tend#1#2{\mathop{\longrightarrow}\limits_{#1\rightarrow#2}}

\def\d{{\partial}}
\def\e{{\rm e}}
\def\eps{\varepsilon}
\def\l{\lambda}

\def\si{{\sigma}}
\def\U{{\mathcal U}}

\numberwithin{equation}{section}

\begin{document}

\title[Nonlinear coherent states in periodic
  potentials]{Nonlinear dynamics of semiclassical coherent states in periodic
  potentials}

\author[R. Carles]{R\'emi Carles}
\address{CNRS \& Univ. Montpellier~2\\ Math\'ematiques 
  CC~051\\ 34095 Montpellier\\ France} 
\email{Remi.Carles@math.cnrs.fr}
\author[C. Sparber]{Christof Sparber}
\address{Department of Mathematics, Statistics, and Computer Science\\
University of Illinois at Chicago\\
851 South Morgan Street
Chicago\\ Illinois 60607, USA}
\email{sparber@uic.edu}

\begin{abstract}
We consider nonlinear Schr\"odinger equations with either local or
nonlocal nonlinearities. In addition, we include  
periodic potentials as used, for example, in matter wave experiments
in optical lattices. By considering the corresponding  
semiclassical scaling regime, we construct asymptotic solutions, which are 
concentrated both in space and in frequency around the effective 
semiclassical phase-space flow induced by Bloch's spectral
problem. The dynamics of these generalized coherent states is governed
by  
a nonlinear Schr\"odinger model with effective mass. In the case of
nonlocal nonlinearities we establish a novel averaging type result in
the critical case.  
\end{abstract} 
\thanks{R.~Carles is supported by the French ANR project
  R.A.S. (ANR-08-JCJC-0124-01). C.~Sparber has been supported by the Royal
Society through his University Research Fellowship.}
\maketitle

\section{Introduction} \label{sec:setting}

Coherent states have been originally introduced in quantum mechanics
to describe wave packets minimizing the uncertainty principle. This
property makes coherent states highly  
attractive for the study of \emph{semiclassical asymptotics}, see, e.g.,
\cite{Hag80, FaGrLi09,  Pau97}. Indeed, it can be shown that for
Schr\"odinger equations with sub-quadratic potentials, coherent states
retain their shape,  
providing minimum uncertainty at all time in the quadratic case
\cite{CoRo06}, and up to Ehrenfest time in general \cite{BoRo02}. Recently,
extensions to weakly nonlinear situations have been studied in 
\cite{CaFe11, CaFe11b}. In addition, the semiclassical dynamics of  
coherent states under the influence of (highly oscillatory) periodic
potentials has been investigated by the authors in \cite{CaSp-p}. In
the present work we \emph{combine}  
the effects coming from periodic and nonlinear potentials. 

To this end, we consider \emph{nonlinear Schr\"odinger equations}
which, after scaling into dimensionless coordinates, appear in the
following semiclassical form: 
\begin{equation}\label{eq:GPE}
  i \eps\d_t \psi^\eps + \frac{\eps^2}{2} \Delta\psi^\eps = V_{\rm per}
     \left(\frac{x}{\eps}\right)\psi^\eps + \eps^\alpha
     f(|\psi^\eps|^2) \psi^\eps, \quad \psi^\eps_{\mid t=0} =  \psi^\eps_0,
 \end{equation}
 where $ t\in \R,\  x\in \R^d$, and $d \in \N$ denotes the spatial
 dimension (usually $d=3$). Moreover, $\eps \in (0,1]$ denotes a
 (small) semiclassical parameter, i.e. a dimensionless rescaled
 Planck's constant.   
 The factor $\eps^\alpha  $ measures the (asymptotic) strength of the
 nonlinearity: the larger the $\alpha>0$, the weaker the nonlinear
 effects.  
 In the following, we shall allow for two different types of
 \emph{gauge invariant nonlinearities}: 
 \begin{itemize}
 \item Local nonlinearities: $f(|\psi^\eps|^2) = \lambda
   |\psi^\eps|^{2\sigma}$, with $\sigma \in \N$ and $\lambda \in \R$,
   allowing for focusing (attractive) and defocusing (repulsive)
   situations. 
 \item Nonlocal nonlinearities of convolution type: $f(|\psi^\eps|^2)
   =  K \ast  |\psi^\eps|^2$, with $K(x)\in \R$ a given interaction
   kernel. 
 \end{itemize}
Finally, the term $V_{\rm per}(x/\eps)$ denotes a highly oscillatory
periodic potential. More precisely, let $\Gamma \simeq \mathbb Z^d$ be
some regular lattice, then  
we assume that for all $y\in \R^d$: $V_{\rm per}(y+\gamma) = V_{\rm
  per}(y)$ with $\gamma \in \Gamma$. In addition, we shall assume
$V_{\rm per} \in C^\infty (\R^d)$.
Equation \eqref{eq:GPE} describes the 
propagation of waves on macroscopic length- and time-scales, i.e. over
many periods of the periodic potential. The parameter $\eps \ll 1$  
consequently describes the ratio between microscopic (quantum
mechanical) and the macroscopic scales.

Nonlinear Schr\"odinger equations with periodic potentials arise in
various physical contexts: A by now classical example is the
mean-field description of electrons propagating in a  
crystalline solid \cite{GuRaTr88, Spohn96} under the additional
influence of a self-consistent electric field. The latter is usually
modeled by means of a nonlocal \emph{Hartree nonlinearity}  
$f(|\psi^\eps|^2) =    |\psi^\eps|^2 \ast 1/|\cdot|$, see, e.g.,
\cite{BMP} for a  
semiclassical study via Wigner measures.
Another situation in which \eqref{eq:GPE} applies is the description of
\emph{Bose-Einstein condensates} in so-called \emph{optical lattices},
cf. \cite{ChNi}. In  
the regime of dilute gases, such condensates can be modeled by the
Gross-Pitaevskii equation with cubic nonlinearity $\sigma = 1$,
cf. \cite{PiSt}. Note however, that other nonlinearities  
also arise, see, for example, \cite{Ber99} where a nonlocal term is
used for the description of superfluid Helium. In addition, strong
magnetic confinement  
allows for the experimental realization of quasi-one-dimensional
(cigar-shaped) condensates, or quasi two-dimensional condensates,  
motivating the fact that we consider \eqref{eq:GPE} in general
dimensions $d\in \N$, cf. \cite{KNSQ}. 
A third example for the appearances of \eqref{eq:GPE} stems from the
description of wave packets propagating within nonlinear  
photonic crystals \cite{photon} and where the nonlinear response of
the media is modeled via a \emph{Kerr nonlinearity} $\sigma = 1$. In
this case, the underlying assumption in the derivation of
\eqref{eq:GPE} is the  
existence of a preferred direction of propagation, implying that the
appropriate model is stated in dimension $d=1$. 

In all of these situations, the joint effects of nonlinearity,
periodicity and dispersion (or, quantum pressure), can lead to the
existence of \emph{stable localized states} conserving the form  
upon propagation and collisions. Gap solitons, discrete breathers and
compactons are examples of such states. Here we shall present another
possibility,  
which will arise from semiclassical description via coherent
states. Even though the nonlinearity in our case is weaker than in the
above mentioned situations  
(due to the fact that $\alpha>0$), the obtained asymptotic solutions
will nonetheless experience nonlinear effects in  
leading order, provided $\alpha$ is of critical size (to be made
precise later on). The latter will depend on the precise form of the
nonlinearity. 

To present our results, we recall the classical \emph{Bloch eigenvalue
  problem} \cite{Wil78}: 
\begin{equation}\label{eq:bloch}
H(k) \chi_m (\cdot,k) = \, E_m (k)\chi_m (\cdot, k), \quad  m \in \N. 
\end{equation}
Denoting by $Y$ the centered fundamental domain of $\Gamma$,
$E_m(k)\in \R$ and $\chi_m(\cdot, k)$ denote, respectively, the
$m$-th eigenvalue/eigenvector pair of 
\begin{equation}\label{eq:bham}
H(k)=  \frac{1}{2} \, \left(- i \nabla_y + k \right)^2+
V_{\rm per}\left (y\right) ,\quad y\in Y,
\end{equation}
parametrized by the \emph{crystal momentum} $k \in  Y^*\simeq \mathbb
T^d$. 
We shall assume that at $t=0$,  
\begin{equation}\label{eq:initial}
  \psi^\eps(0,x)\Eq \eps 0\eps^{-d/4} u_0 
\left(\frac{x-q_0}{\sqrt\eps}\right) \chi_m\left(\frac{x}{\eps},
  p_0 \right) \e^{i p_0 \cdot (x-q_0)/ \eps}, 
\end{equation}
where $u_0$ denotes some \emph{smooth and rapidly decaying
  profile}. In other words, the initial data $\psi_0^\eps$ can be
approximated by  
a highly oscillatory Bloch eigenfunction $\chi_m$ modulated by a
(generalized) coherent state, i.e.  
a wave function which is localized both in space and in frequency. 
\begin{remark} 
In particular, the choice $u_0(z)=\exp(-|z|^2/2)$ 
yields a classical coherent state, i.e. ground state of the harmonic
oscillator potential, as modulation.  
The states \eqref{eq:initial} are more general, though, since we can
allow for any $u_0 \in \mathcal S(\R^d)$, the Schwartz space of
rapidly decaying, smooth functions.  
We also remark that the same class of initial data have recently been
considered in \cite{CaSp-p}, where the situation of linear
Schr\"odinger equations with combined  
periodic and slowly varying external potentials has been considered
(see \ref{sec:appB} for more details).  
\end{remark}
Provided such initial data, we shall show that the solution of
\eqref{eq:GPE} can be approximately (in a sense to be made precise)
described by the following \emph{semiclassical wave packet}: 
\begin{equation}
  \label{eq:semibl}
 \psi^\eps(t,x)\Eq \eps 0 \varphi^\eps(t,x):=\eps^{-d/4} u 
\left(t,\frac{x-q(t)}{\sqrt\eps}\right) \chi_m\left(\frac{x}{\eps},
  p_0 \right) e^{i \phi_m(t,x) / \eps}
\end{equation}
where $q(t)= q_0 + t \nabla_k E_m(p_0)$ describes the
\emph{macroscopic shift} of the centre of mass and the highly
oscillatory phase function $\phi_m$ is   
\begin{equation}\label{eq:newphase}
\phi_m(t,x) =  p_0 \cdot (x- q_0 ) - t   E_m(p_0).
\end{equation}
To this end, we need to give sense to the \emph{group velocity}
$\nabla_k E_m(k)$ and thus, we have to assume from now on, that: 
\begin{equation}\label{ass:simple} 
\mbox{$E_m(k)$ is a \emph{simple eigenvalue} in the vicinity of $k = p_0$}. 
\end{equation}
(Of course, $p_0\in \R^d$ has to be understood modulo $\Gamma^*$ in this case.) 
In other words, we have to avoid that two Bloch bands cross at $p_0$,
that is $E_m(p_0) = E_n(p_0)$, for $m\not =n$. It is known that at
such crossing points $E_m(k)$ is no longer  
differentiable, causing the above asymptotic
description (which is based on an adiabatic decoupling of the slow and
fast degrees of freedom)  
to break down. 
\begin{remark} Clearly, the non-crossing condition given above
  restricts our choice for the initial wave vectors $ p_0 \in \R^d $. 
It is known however 
that the set of band crossings has Lebesgue measure zero. 
For example, in the case $d=1$, band crossings can only occur at $k=0$
or at the boundary of the Brillouin zone. 
\end{remark}
So far, we have not said what determines the profile $u=u(t,z)$
appearing in \eqref{eq:semibl}. Its time-evolution depends on the
strength of the nonlinearity, i.e. on the size of $\alpha >0$  
in the case of local nonlinearities (the situation for nonlocal ones
will be described later on).  
We shall find that the \emph{critical size} is $\alpha_c= 1+d\si/2$
and when $\alpha=\alpha_c$, 
$u$ solves the following \emph{homogenized nonlinear Schr\"odinger equation} 
\begin{equation}\label{eq:NLprofile}
  i\d_t u + \frac{1}{2}{\rm div}_z\((\nabla^2_k E_m\(p_0\))\cdot
  \nabla_z \) u = \lambda_m |u|^{2\sigma} u , \quad u_{\mid t=0} =  u_0,
\end{equation}
with effective coupling constant
\begin{equation*}
  \lambda_m  =  \lambda \int_Y\left| \chi_m\(y,p_0\)\right|^{2\si+2}dy.
\end{equation*}
Note that the dispersive properties of \eqref{eq:NLprofile} are determined by 
an \emph{effective mass matrix} $\nabla^2_k E_m\(p_0\)\in \R^{d\times d}$,
which itself depends on the choice of the initial wave vector
(cf. \cite{DSXZ} for a recent study).  

In the next section, we  derive this effective mass equation
from multi-scale expansion.  
A rigorous stability result is then  proved in Section
\ref{sec:main}. The case of nonlocal nonlinearities is treated in
Section \ref{sec:nonlocal}. As we shall see, in situations where the
kernel $K$ is homogeneous, the critical 
value $\alpha_c$ depends on the degree of homogeneity (like in the
case of a local nonlinearity), and the analogue of 
\eqref{eq:NLprofile} is an envelope equation with a nonlocal
nonlinearity. On the other hand, if the kernel $K$ is a smooth function, then
$\alpha_c=1$, and in sharp contrast with the other situations studied in this
paper, the analogue of \eqref{eq:NLprofile} for $\alpha=\alpha_c$ is
found to be a  
\emph{linear} equation, see Section~\ref{sec:smoothK}. 

\section{Multi-scale expansion}\label{sec:multi}

\subsection{The hierarchy of equations}
\label{sec:hierarchy}

Except for the treatment of the nonlinear term, we resume the strategy
followed in \cite{CaSp-p}. 
We seek the solution $\psi^\eps$ of \eqref{eq:GPE} in the form
\begin{equation}\label{eq:ansatz}
  \psi^\eps(t,x)= \eps^{-d/4}\,
  \U^\eps\(t,\frac{x-q(t)}{\sqrt\eps},\frac{x}{\eps}\) \e^{i \phi_m(t,x) /
    \eps}, 
\end{equation}
where the phase $\phi_m(t,x)$ is given by \eqref{eq:newphase}  and the function
$\U^\eps=\U^\eps(t,z,y)$ admits an asymptotic expansion
\begin{equation}\label{eq:DA}
  \U^\eps(t,z,y)\Eq \eps 0 \sum_{j\in \N}\eps^{j/2}U_j(t,z,y) 
\end{equation}
with smooth profiles $U_j$ which, in addition, are assumed to be
$\Gamma$-periodic 
with respect to $y$. For $\psi^\eps$ given by the ansatz
\eqref{eq:ansatz}, we compute  
\[
 \(i\eps \d_t
  \psi^\eps+\frac{\eps^2}{2}\Delta \psi^\eps - V_{\rm per}
  \(\frac{x}{\eps}\) \psi^\eps \) = \eps^{-d/4}\e^{i\phi_m/\eps}
  \sum_{j=0}^2 \eps^{j/2}b_j^\eps(t,z,y)
\Big|_{(z,y)=\(\frac{x-q(t)}{\sqrt\eps},\frac{x}{\eps}\)}
\]
with 
\begin{align*}
 b_0^\eps&= -\d_t \phi_m \U^\eps  +\frac{1}{2} \Delta_y \U^\eps
-\frac{1}{2}| p_0|^2\U^\eps 
+i p_0\cdot \nabla_y \U^\eps
-V_{\rm per}(y) \U^\eps
,\\
 b_1^\eps&= - i \dot q(t) \cdot \nabla_z \U^\eps
+\(\nabla_y\cdot\nabla_z\)\U^\eps + i p_0 \cdot \nabla_z \U^\eps,\\
 b_2^\eps&= i \d_t \U^\eps +\frac{1}{2}\Delta_z \U^\eps.
\end{align*}
Using \eqref{eq:newphase} and the fact that $\dot q(t) = \nabla_k
E_m(p_0)$, we can rewrite  
$b_0^\eps, b_1^\eps$ as  
\begin{align*}
  b_0^\eps&=  (E_m(p_0) -  H\(p_0\) ) \U^\eps  ,\\
 b_1^\eps &= i \(p_0 -\nabla_k E_m\(p_0\)\)\cdot \nabla_z \U^\eps + 
\(\nabla_y\cdot\nabla_z\)\U^\eps,
\end{align*}
where $H(p_0)$ is the Bloch Hamiltonian \eqref{eq:bham} evaluated at
$k=p_0$ (again, this has to be understood modulo $\Gamma^*$). 
Introducing the following linear operators:
\begin{equation*}
 L_0 = E_m\(p_0\) -H \(p_0\),\ L_1 = i \(p_0 -\nabla_k
 E_m\(p_0\)\)\cdot \nabla_z +  
\nabla_y\cdot\nabla_z,\ L_2= i \d_t  +\frac{1}{2}\Delta_z  ,
\end{equation*}
and expanding $\U^\eps$ in powers of $\eps \in (0,1]$, we consequently
need to solve the following hierarchy of equations: 
\begin{equation}
\label{eq:systU}
\left \{
  \begin{aligned}
    &L_0 U_{0}=0,\\
&L_0 U_{1}+L_1U_{0}=0,\\
&L_0 U_{2}+L_1 U_{1}+L_2 U_{0}=F(U_0),
  \end{aligned}
\right.
\end{equation}
where, for $\alpha_c =1+d\sigma/2$, we find:
\begin{equation*}
  F(U_0)=\left\{
    \begin{aligned}
      0 & \quad{\rm if }\ \ \alpha>\alpha_c,\\
\l |U_0|^{2\si}U_0 & \quad {\rm if }\ \ \alpha=\alpha_c.
    \end{aligned}
\right.
\end{equation*}
In the next subsection, we shall focus on the resolution of \eqref{eq:systU}.

\subsection{The effective mass equation}
\label{sec:homo}

Given the form of $L_0$, the equation $L_0U_0=0$ implies
\begin{equation}\label{eq:U0}
  U_0(t,z,y)= u(t,z)\chi_m\(y,p_0\). 
\end{equation}
By Fredholm's alternative, a necessary and sufficient condition to
solve the equation $L_0U_1+L_1U_0=0$, is that $L_1U_0$ is orthogonal to
$\ker L_0$, that is:
\begin{equation}\label{eq:orth1}
  \<\chi_m,L_1 U_0\>_{L^2(Y)}=0. 
\end{equation}
Given the expression of $L_1$ and the formula \eqref{eq:U0}, we
compute
\begin{equation*}
  L_1U_0 = i \(p_0-\nabla_kE_m\(p_0\)\)\cdot \nabla_z u(t,z)
  \chi_m\(y,p_0\) + \nabla_y\chi_m\(y,p_0\)\cdot \nabla_z u(t,z). 
\end{equation*}
Now, we make use of the algebraic identities derived in Section~\ref{sec:appA}. 
In view of \eqref{eq:dkE}, we infer that \eqref{eq:orth1} is
automatically fulfilled. We thus obtain
\begin{equation*}
  U_1(t,z,y) = u_1(t,z)\chi_m\(y,p_0\) + u_1^\perp(t,z,y),
\end{equation*}
where $u_1^\perp$, the part of $U_1$ which is orthogonal to $\ker
L_0$, is obtained by inverting an elliptic operator: 
$  u_1^\perp = -L_0^{-1}L_1 U_0$.
Note that the formula for $L_1U_0$ can also be written as
\begin{equation*}
  L_1U_0 =-i  \nabla_k \(E_m\(p_0\)-H\(p_0\)\)
  \chi_m\(y,p_0\) \cdot \nabla_z u(t,z).
\end{equation*}
Taking into account \eqref{eq:dbloch}, this yields: $ u_1^\perp
(t,z,y) = -i \nabla_k \chi_m\(y,p_0\)\cdot \nabla_z 
  u(t,z).$ At this stage, we shall, for simplicity, choose $u_1=0$, in
  which case 
$U_1$ becomes 
simply a function of $u$:
\begin{equation}\label{eq:U1}
  U_1(t,z,y) = -i \nabla_k \chi_m\(y,p_0\)\cdot \nabla_z
  u(t,z). 
\end{equation}
As a next step in the formal analysis, we must solve
$  L_0 U_2+L_1 U_1+L_2 U_0=F(U_0)$.
By the same argument as before, we require
\begin{equation}\label{eq:orth2}
  \<\chi_m,L_1 U_1+L_2 U_0-F(U_0)\>_{L^2(Y)}=0. 
\end{equation}
With the expression \eqref{eq:U1}, we compute
\begin{equation*}
  L_1U_1 = \sum_{j,\ell=1}^d \( \(p_0-\nabla_k E_m\(p_0\)\)_j
  \d_{k_\ell}\chi_m\(y,p_0\) -i \d_{y_j
    k_\ell}^2\chi_m\(y,p_0\)\)\d_{z_j z_\ell}^2 u,
\end{equation*}
and we also have
\begin{equation*}
 L_2U_0 = \( i \d_t +\frac{1}{2}\Delta_z \) u(t,z) \, \chi_m\(y,p_0\) .
\end{equation*}
We consequently infer
  \begin{align*}
   \<\chi_m,L_1 U_1+L_2 U_0\>_{L^2(Y)} &=    i \d_t u +\frac{1}{2}\Delta_z u \\
 -\sum_{j,\ell=1}^d &\< \chi_m,\d_{k_j} E_m\(p_0\)
  \d_{k_\ell}\chi_m +i \d_{y_j
    k_\ell}^2\chi_m\>_{L^2(Y)}\d_{z_j z_\ell}^2 u.
  \end{align*}
 In the case $\alpha = \alpha_c$, by making the last sum symmetric
 with respect to $j$ and $\ell$, and 
using \eqref{eq:d2kE}, we finally obtain the nonlinear Schr\"odinger 
equation \eqref{eq:NLprofile} with effective mass tensor $M= \nabla_k^2
E_m(p_0)$ and coupling constant 
\begin{equation*}
  \lambda_m : = \l \<\chi_m, |\chi_m|^{2\sigma} \chi_m
  \>_{L^2(Y)}=\lambda \int_Y\left| \chi_m\(y,p_0\)\right|^{2\si+2}dy. 
\end{equation*}
In addition, we can write
\begin{equation}\label{eq:U2}
  U_2(t,z,y)= u_2(t,z)\chi_m\(y,p_0\) + u_2^\perp(t,z,y),
\end{equation}
where $ u_2^\perp = -L_0^{-1}\(L_1 U_1+L_2 U_0\).$
We shall also impose $u_2\equiv 0$ and thus $U_2 = u_2^\perp$. 

Assume for the moment that we can solve
\eqref{eq:NLprofile}. Then, we have the following result, which
establishes some  
basic regularity properties of our multi-scale expansion (where we
denote by $H^k$ the usual $L^2(\R^d)$ based Sobolev space). 

\begin{lemma}\label{lem:app} Suppose \eqref{ass:simple} holds true and
  let $u\in C([0,T]; H^k)$, be a solution of \eqref{eq:NLprofile} up
  to some $T>0$.  
Then $U_j \in C([0,T]; H^{k-j}_z\times W^{\infty, \infty}(Y))$, for $j=0,1,2$. 
\end{lemma}

\begin{proof} First note that $(y,k)\mapsto \chi_m(y,k)$ is 
smooth and bounded together with all its derivatives, provided
\eqref{ass:simple} holds true.  
Having this in mind, the proof follows directly from the construction
of $\{U_j\}_{j=0,1,2}$ as solutions to the system  
\eqref{eq:systU}.
\end{proof}

\begin{remark} 
Note that in the case $\alpha > \alpha_c$ nonlinear effects are
\emph{absent at leading order} since we obtain, instead of
\eqref{eq:NLprofile}, a {\it linear} effective mass equation: 
\begin{equation}\label{eq:lineq}
 i \d_t u_{\rm lin} + \frac{1}{2}{\rm div}_z\(\nabla^2_k E_m\(p_0\)\cdot
  \nabla_z\) u_{\rm lin} = 0, \quad u_{\mid t=0} =  u_0 .
\end{equation}
This type of equation has been derived in \cite{AlPi05,Sp}, using a different 
asymptotic scaling.
\end{remark}

\section{Main results}
\label{sec:main}

In this section we shall make the computations given above rigorous
and prove a nonlinear stability result.  
As a first step we need to guarantee the existence of a smooth
solution to \eqref{eq:NLprofile}, at least locally in-time.

\subsection{Construction of an approximate solution}\label{sec:approx} 

The dispersion relation of \eqref{eq:NLprofile} is given by a
real-valued symmetric matrix. Standard techniques (see, e.g.,
\cite{TaoDisp}) yield the 
existence of a unique local solution, provided that the initial datum
is sufficiently smooth:
\begin{lemma}\label{lem:exist}
  Let $u_0\in H^k$ with $k>d/2$. There exists $T_c\in (0,+\infty]$
  and a unique maximal solution $u\in C([0,T_c);H^k)$ to
  \eqref{eq:NLprofile}, such that $\| u(t,\cdot) \|_{L^2} = \| u_0
  \|_{L^2}$. The solution is maximal  
  in the sense that if $T_c<\infty$, then
  \begin{equation*}
    \lim_{t\to T_c}\|u(t,\cdot )\|_{H^k}=+\infty. 
  \end{equation*}
\end{lemma}
The solution $u(t,\cdot)$ may not exist for all times, even if
$\lambda \ge 0$, i.e. even if the nonlinearity in the original
equation \eqref{eq:GPE} is defocusing.  
However, we can claim $T_c=\infty$ in either of the
following cases (see e.g. \cite{CazCourant}):
\begin{itemize}
\item $\nabla_k^2E_m(p_0)$ is positive definite and $\lambda_m\ge 0$,
  or
\item $\nabla_k^2E_m(p_0)$ is negative definite and $\lambda_m\le 0$.
\end{itemize}
On the contrary, if for instance $\nabla_k^2E_m(p_0)$ is positive
definite and $\lambda_m< 0$ (focusing nonlinearity), finite time blow
up (that is, $T_c<\infty$) may occur, see, e.g., \cite{TaoDisp,CazCourant}. This
is the case typically if the initial datum is ``too large'': for any
fixed profile $u_0\in {\mathcal S}(\R^d)$, if one considers $u_{\mid
  t=0} = \Lambda u_0$, there exists $\Lambda_0>0$ such that for all 
$\Lambda\ge \Lambda_0$, $T_c<\infty$. 
Note that in other situations, where the signature of
$\nabla_k^2E_m(p_0)$ is non-trivial (hence $d\ge 2$), the 
issue of global existence \emph{vs.} finite time blow-up is an open question. 

\begin{remark} Clearly, for $\alpha > \alpha_c$ these issues do not
  occur, since the leading order profile $u_{\rm lin} $  
solves the linear equation \eqref{eq:lineq} and hence exists for all
times, $T_c = +\infty$. 
\end{remark}

Lemma \ref{lem:exist} provides the existence of a local-in-time
solution $u$ of the effective mass equation. In view of the
multi-scale expansion  
given in Section~\ref{sec:multi}, we can thus define an approximate solution by
\begin{equation}
  \label{eq:app}
   \psi_{\rm app}^\eps(t,x):= \eps^{-d/4}
  \(U_0 + \sqrt{\eps} U_1 + \eps U_2 \)
  \(t,\frac{x-q(t)}{\sqrt\eps},\frac{x}{\eps}\) \e^{i \phi_m(t,x) / 
    \eps},
\end{equation}
which satisfies the original equation \eqref{eq:GPE} up to some remainder terms:
\begin{align*}
   \(i\eps \d_t
   +\frac{\eps^2}{2}\Delta  -
  V_{\rm per} 
   \)&\psi_{\rm app}^\eps-\l \eps^\alpha |\psi_{\rm
    app}^\eps|^{2\si}\psi_{\rm app}^\eps= \\
 &\frac{\e^{i\phi_m/\eps}}{\eps^{d/4}}
 \eps^{3/2}\(r_1^\eps+r_2^\eps\) (t,z,y)\Big|_{(z,y)= 
    \(\frac{x-q(t)}{\sqrt\eps},\frac{x}{\eps}\) }-\l R^\eps(t,x).
  \end{align*}
The remainder terms are given by 
\begin{equation*}
 r_1^\eps(t,z,y)=L_2U_1(t,z,y),\\ r_2^\eps(t,z,y)=
 L_1U_2(t,z,y)+\sqrt{\eps} L_2U_2(t,z,y), 
\end{equation*}
and $R^\eps = \eps^\alpha |\psi_{\rm
    app}^\eps|^{2\si}\psi_{\rm app}^\eps$ if $\alpha>\alpha_c$, while
  if $\alpha=\alpha_c=1+d\si/2$,
\begin{equation*}
 R^\eps=
 \eps^{1+d\si/2} |\psi_{\rm
    app}^\eps|^{2\si}\psi_{\rm app}^\eps-\frac{\eps}{\eps^{d/4}} \e^{i\phi_m/\eps}
  \left|U_0\(t,\frac{x-q(t)}{\sqrt\eps},\frac{x}{\eps}\)
\right|^{2\si}U_0\(t,\frac{x-q(t)}{\sqrt\eps},\frac{x}{\eps}\). 
\end{equation*}
This, together with the regularity result established in Lemma
\ref{lem:app} then directly yields the following proposition.

\begin{proposition}\label{prop:app} 
Assume \eqref{ass:simple} and let $\alpha \ge
\alpha_c=1+d\si/2$. Then, we can find $\psi_{\rm app}^\eps$ such 
that:\\
$1.$ For all $T\in [0, T_c)$,
$\psi_{\rm app}^\eps$ has a coherent state 
structure on $[0,T]$, in the sense that there exists $C$ independent
of $\eps\in (0,1]$ such that, with  $\varphi^\eps$ defined in
\eqref{eq:semibl},  
\begin{align*}
\text{ for } \alpha=\alpha_c, &\quad  \sup_{t\in
  [0,T]}\|\psi_{\rm app}^\eps(t, \cdot)-\varphi^\eps(t,
\cdot)\|_{L^2(\R^d)}\le C\sqrt\eps, \\ 
 \text{ for } \alpha>\alpha_c,&\quad  \sup_{t\in
  [0,T]}\|\psi_{\rm app}^\eps(t, \cdot)-\varphi_{\rm 
    lin}^\eps(t, \cdot)\|_{L^2(\R^d)}\le C\sqrt\eps,
\end{align*}
where $\varphi_{\rm
  lin}^\eps$ is the approximate solution constructed from $u_{\rm lin}$,
solving \eqref{eq:lineq}.\\ 
$2.$ The function solves \eqref{eq:GPE} up to a small error:
\begin{equation*}
     i\eps\d_t \psi_{\rm app}^\eps +\frac{\eps^2}{2}\Delta
\psi_{\rm app}^\eps=V_{\rm per} 
     \left(\frac{x}{\eps}\right)\psi_{\rm app}^\eps+
     \lambda\eps^\alpha|\psi_{\rm app}^\eps|^{2\sigma}\psi_{\rm
       app}^\eps+\eps w^\eps,  
\end{equation*}
where the remainder term $w^\eps$ satisfies: for all $T>0$, with
$T<T_c$ in the case $\alpha=\alpha_c$, there
exists $C>0$ independent of $\eps>0$ such that
\begin{equation*}
  \sup_{t\in [0,T]}\|w^\eps(t, \cdot)\|_{L^2(\R^d)}\le C
\left\{
\begin{aligned}
\eps^{\min(\alpha-\alpha_c,1/2)}&\quad \text{ if } \alpha>\alpha_c,\\
\sqrt\eps &\quad \text{ if } \alpha=\alpha_c.
\end{aligned}
\right.
\end{equation*}
\end{proposition}
Note that because of the factor $\eps$ in front of the time
derivative, it is natural to represent a small error term as $\eps$
times as small term. 

\subsection{Nonlinear stability} \label{sec:stab}

It remains to prove nonlinear stability of the approximate solution
constructed above. For the sake of simplicity we shall do so only for
$\alpha=\alpha_c$ and  
$d=1$. The (physically less interesting) case $\alpha>\alpha_c$ can be
proved analogously and a possible generalization to higher dimensions
is indicated in Remark~\ref{rem:multi} below.  
For $\eps_0>0$, set
\begin{equation*}
 \|f^\eps\|_{H^1_\eps}:=\sup_{0<\eps\le 
      \eps_0}\big(\|f^\eps\|_{L^2}+\|\eps\d_x f^\eps \|_{L^2}\big),
      \end{equation*}
which is equivalent to the usual $H^1$-norm for every (fixed) $\eps>0$.      
The approach that we present is similar to the one followed in
\cite{CaMaSp04}:
First, we need to construct a more accurate approximate
  solution than the one stated in Proposition~\ref{prop:app}. Taking
  the asymptotic expansion presented in Section~\ref{sec:multi} one
  step further, we can 
  gain a factor $\sqrt\eps$ in Proposition~\ref{prop:app}. More
  precisely, we can  
  construct $\widetilde \psi^\eps_{\rm app}$ such that:
  \begin{equation}\label{eq:psiapp2}
    \sup_{t\in [0,T]}\left\| \psi^\eps_{\rm app}(t,\cdot)-\widetilde
      \psi^\eps_{\rm app}(t,\cdot)\right\|_{H^1_\eps} \le C\sqrt\eps,
  \end{equation}
 and 
  \begin{equation*}
        i\eps\d_t \widetilde\psi_{\rm app}^\eps +\frac{\eps^2}{2}\Delta
\widetilde\psi_{\rm app}^\eps=V_{\rm per} 
     \left(\frac{x}{\eps}\right)\widetilde\psi_{\rm app}^\eps+
     \lambda\eps^{1+\si/2}|\widetilde\psi_{\rm
       app}^\eps|^{2\sigma}\widetilde\psi_{\rm 
       app}^\eps+\eps \widetilde w^\eps,  
  \end{equation*}
where the additional factor $\sqrt\eps$ is reflected in the error estimate
\begin{equation}\label{eq:wtilde}
  \sup_{t\in [0,T]}\left\|\widetilde w^\eps(t,\cdot)\right\|_{H^1_\eps}\le C
  \eps. 
\end{equation}
Note that in this case the corrector $U_1$ is not the same for $\widetilde
\psi^\eps_{\rm app}$, since unlike what we have done in
\S\ref{sec:homo}, we can no longer assume $u_1=0$. Rather, $u_1$ now solves an
evolution equation, which is essentially \eqref{eq:NLprofile} linearized about
$u$, with a non-trivial source term (see \cite{CaMaSp04} for more
details). Therefore, the estimate 
\eqref{eq:psiapp2} must be expected to be sharp in general. 

Having constructed such an improved approximation $\widetilde
\psi^\eps_{\rm app}$ we can state the following nonlinear stability result:
\begin{theorem} \label{thm:stab} 
Let $d=1$, $\alpha =1+\si/2$, $\si\in \N$, and
Assumption~\eqref{ass:simple} hold. In  
addition, suppose that the initial data satisfy:
\begin{equation}\label{eq:wp}
 \left\|\, \psi_0^\eps - \widetilde\psi^\eps_{{\rm app}\mid t=0}\right\|_{L^2(\R)}
  =\O(\eps),  \quad \left\|\, \eps\d_x\(\psi_0^\eps - \widetilde\psi^\eps_{{\rm
        app}\mid t=0}\)\right\|_{L^2(\R)} =\O(1). 
\end{equation}
Let $T\in [0,T_c)$. Then, there exists $\eps_0=\eps_0(T)$ such that for
$\eps\in (0,\eps_0]$, the solution of \eqref{eq:GPE} exists on $[0,T]$.
Moreover,
there exists $C$ independent of $\eps\in (0,\eps_0]$ such that
\[
\sup_{t\in [0,T]} \| \psi^\eps(t,\cdot) - \varphi^\eps
(t,\cdot) \|_{L^2(\R)} \le C\sqrt{\eps}. 
\]
where $\varphi^\eps$ is defined in \eqref{eq:semibl}.
\end{theorem}

\begin{proof} The scheme of the proof is the same as the proof of
  Theorem~4.5 in \cite{CaMaSp04}, so we shall only give the main
  steps. 
Fix $T<T_c$ and let $\eta^\eps=\psi^\eps -\widetilde \psi^\eps_{\rm
  app}$ be the error 
  between the exact and the approximate solution. It
  satisfies 
\begin{equation*}
       i\eps\d_t \eta^\eps +\frac{\eps^2}{2}\Delta
\eta^\eps=V_{\rm{per}} 
     \left(\frac{x}{\eps}\right)\eta^\eps+ 
     \lambda\eps^{1+\si/2}\(|\psi^\eps|^{2\si}\psi^\eps- |\widetilde\psi_{\rm
       app}^\eps|^{2\sigma}\widetilde\psi_{\rm 
       app}^\eps\)-\eps \widetilde w^\eps, 
\end{equation*}
with $\|\eta^\eps_{\mid t=0}\|_{L^2}=\O(\eps)$,
$\|\eps\d_x\eta^\eps_{\mid t=0}\|_{L^2}=\O(1)$ by assumption.  
From \cite{RauchUtah}, we have:
\begin{lemma}[Moser's lemma]\label{lem:moser}
Let $R>0$, $s\in \N$ and $F(z)=|z|^{2\si}z$, $\si\in \N$. Then there exists
$C=C(R,s,\si)$ such that if $v^\eps$ satisfies 
\begin{equation*}
\left\| \(\eps\d_x\)^\beta v^\eps\right\|_{L^\infty(\R)} \le R,
\quad 0\le \beta\le s ,
\end{equation*}
and $\delta^\eps$ satisfies $\displaystyle \left\| \delta^\eps
\right\|_{L^\infty(\R)} \le 
R$, then 
\begin{equation*}
 \sum_{0\le \beta\le s}\left\|(\eps\d_x)^\beta   \(F(v^\eps+\delta^\eps)
- F(v^\eps)\)\right\|_{L^2} 
\le C \sum_{0\le \beta\le s}\left\| \(\eps\d_x\)^\beta\delta^\eps\right\|_{L^2} . 
\end{equation*}
\end{lemma}
We apply this lemma with $v^\eps=\eps^{1/4}\widetilde\psi^\eps_{\rm
  app}$, and $s=0$, $s=1$ successively: there exists 
$R>0$ independent of $\eps\in (0,1]$ such that  
\begin{equation*}
\sup_{t\in [0,T]}  \sum_{\beta=0,1} \left\| \(\eps\d_x\)^\beta
   v^\eps(t)\right\|_{L^\infty(\R)} \le R. 
\end{equation*}
Set $\delta^\eps = \eps^{1/4}\eta^\eps$. By assumption and the 
Gagliardo--Nirenberg inequality,
\begin{equation}\label{eq:GN}
  \|\delta^\eps_{\mid t=0}\|_{L^\infty} =\eps^{1/4}\|\eta^\eps_{\mid
    t=0}\|_{L^\infty} \le \eps^{1/4}\sqrt 2\eps^{-1/2}\|\eta^\eps_{\mid
    t=0}\|_{L^2}^{1/2}\|\eps\d_x\eta^\eps_{\mid
    t=0}\|_{L^2}^{1/2}\le C\eps^{1/4}. 
\end{equation}
As long as $\|\delta^\eps(t)\|_{L^\infty}\le R$, energy estimates and
Moser's lemma with $s=0$ yield
\begin{equation*}
  \|\eta^\eps(t)\|_{L^2}\le \|\eta^\eps(0)\|_{L^2} + C \int_0^t
  \|\eta^\eps(s)\|_{L^2}ds + \int_0^t\|\widetilde w^\eps(s)\|_{L^2}ds,   
\end{equation*}
where we have used the homogeneity of
$F$. By Gronwall's Lemma, for $t\le T$:
\begin{equation*}
  \|\eta^\eps(t)\|_{L^2}\le C(T) \(\|\eta^\eps(0)\|_{L^2} + \int_0^t
  \|\widetilde w^\eps(s)\|_{L^2}ds\)\le C\eps. 
\end{equation*}
Applying the operator $\eps\d_x$ to the equation satisfied by
$\eta^\eps$, we infer similarly
\begin{align*}
 \|\eps\d_x\eta^\eps(t)\|_{L^2}\le & \,
   \|\eps\d_x\eta^\eps(0)\|_{L^2} + C \int_0^t 
  \|\eta^\eps(s)\|_{H^1_\eps}ds +
  \int_0^t\|\eps\d_x\widetilde w^\eps(s)\|_{L^2}ds\\
& \, +
  \frac{1}{\eps}\|\partial_y V_{\rm{per}}\|_{L^
    \infty}\int_0^t\|\eta^\eps(s)\|_{L^2}ds  ,  
\end{align*}
where the last term stems from the relation
$[\eps\d_x,V_{\rm{per}}(x/\eps)]= \d_y V_{\rm{per}} (x/\eps)\in
L^\infty$, since $V_{\rm{per}}$ is smooth and periodic. Thus,
\begin{equation*}
 \|\eps\d_x\eta^\eps(t)\|_{L^2}\le C + C  \int_0^t
 \|\eps\d_x\eta^\eps(s)\|_{L^2}ds + Ct. 
\end{equation*}
Gronwall's lemma now yields $
  \|\eps\d_x\eta^\eps(t)\|_{L^2}\le C(T)$. In view of the
  Gagliardo--Nirenberg inequality, 
\begin{equation*}
 \|\delta^\eps(t)\|_{L^\infty} = \eps^{1/4}\|\eta^\eps(t)\|_{L^\infty}
   \le \sqrt 2
  \eps^{-1/4}\|\eta^\eps\|_{L^2}^{1/2}\|\eps\d_x\eta^\eps\|_{L^2}^{1/2}\le
  C(T)\eps^{1/4}. 
\end{equation*}
For $\eps$ sufficiently small (depending of $T$),
$\|\delta^\eps(t)\|_{L^\infty}\le R$ for all $t\in [0,T]$, and the
result follows from a bootstrap argument. 
 \end{proof}
 The above theorem shows nonlinear stability of the approximate
solution up to times of order $\mathcal O(1)$, i.e. independent of
$\eps$, provided that the initial data are \emph{well-prepared}, in
the sense given in 
\eqref{eq:wp}. Essentially this means that $\psi_0^\eps$ contains not
only $U_0$, but 
also $\widetilde U_1$ associated to $\widetilde \psi^\eps_{\rm
  app}$. We shall not insist further on this aspect, which is probably
a technical artifact, and remark that in the linear case a stronger
result is valid, see  
\cite{CaSp-p}   
where stability is proved up to the so-called \emph{Ehrenfest time} $\mathcal
O(\ln 1/\eps)$, and no well-preparedness as in \eqref{eq:wp} is needed
(an initial error $\O(\eps^r)$ for \emph{some} $r>0$ suffices). 

 \begin{remark}\label{rem:multi}
   If $x\in \R^d$ with $d\ge 2$, the proof can be easily adapted,
   provided an even better approximate solution is constructed. The
   reason is that, instead of \eqref{eq:GN},  
   one needs to rely on the following 
   Gagliardo--Nirenberg inequality 
   \begin{equation*}
  \|\delta^\eps \|_{L^\infty(\R^d)} \le C\eps^{-d/2}
  \|\delta^\eps \|_{L^2(\R^d)}^{1-d/(2s)}\, \|
 \, |\eps\nabla|^s\delta^\eps\|_{L^2(\R^d)}^{d/(2s)}, \quad \mbox{for $s > d/2.$} 
\end{equation*}
   Thus, in order to account for the singular factor $\eps^{-d/2}$,
   one is forced to  
   construct an approximate solution $\widetilde \psi_{\rm app}^\eps$
   to a sufficiently high order in $\eps$ (see \cite{CaMaSp04} for
   more details).  
 \end{remark}

\section{The case of nonlocal nonlinearities}
\label{sec:nonlocal}

In this section we shall show how to perform the same asymptotic
analysis as before in the case of nonlocal nonlinearities. In other
words, we consider  
\begin{equation}\label{eq:r3}
    i\eps\d_t \psi^\eps +\frac{\eps^2}{2}\Delta \psi^\eps=V_{\rm per}
     \left(\frac{x}{\eps}\right)\psi^\eps +
     \eps^\alpha (K\ast |\psi^\eps|^2)\psi^\eps
\end{equation}
with $K(x)\in \R$ some given interaction kernel. In the following we
shall focus on two particular choices of interaction kernels $K$ which
are physically relevant.

\subsection{Homogeneous kernels} \label{sec:homK} In this subsection
we shall consider functions of the form 
\[ 
K(x)= \lambda |x|^\mu, \ \lambda \in\R,\ \mbox{with $\mu\in
  \R\setminus\{0\}$ such that}\ 
-\min(2,d)<\mu \le 2.\] 
For example, the choice $\mu=-1$ in $d=3$ corresponds to the classical
Hartree nonlinearity, modeling a self-consistent, repulsive
($\lambda>0$) Coulomb interaction. The case $\mu>0$ has been recently
studied in \cite{Ma-p}. 

Like in the case of local nonlinearities, the critical exponent
$\alpha_c$ depends on 
the homogeneity $\mu$, namely $\alpha_c = 1-\mu/2$.  
This can be seen as follows:  We plug the ansatz
\begin{equation*}
  \varphi^\eps(t,x)=\eps^{-d/4} u 
\left(t,\frac{x-q(t)}{\sqrt\eps}\right) \chi_m\left(\frac{x}{\eps},
  p_0 \right) \e^{i \phi_m(t,x) / \eps}
\end{equation*}
into the convolution term $\eps^{1-\mu/2}\(|x|^\mu
\ast|\psi^\eps|^2\)$. This yields  
\begin{equation*}
 \eps^{1-\mu/2-d/2}\int_{\R^d} |x-\xi|^\mu \left|
  u\left(t,\frac{\xi-q(t)}{\sqrt\eps}\right)\right|^2 \left|
 \chi_m\left(\frac{\xi}{\eps},
  p_0 \right)  \right|^2d\xi.
\end{equation*} 
We want this term to be of order $\O(\eps)$ in our asymptotic
expansion, to mimic the approach presented in \S\ref{sec:multi}. In
this case, it will consequently appear within $b_2^\eps$,  
leading to the effective mass equation. In order to show that this is
indeed the case, we rewrite the initial convolution as 
\begin{equation*}
\eps^{1-\mu/2-d/2}\int_{\R^d} |\xi|^\mu \left|
  u\left(t,\frac{x-\xi-q(t)}{\sqrt\eps}\right)\right|^2 \left|
  \chi_m\left(\frac{x-\xi}{\eps}, p_0\right) \right|^2d\xi,
\end{equation*}
and use the substitution $z= (x-q(t))/\sqrt\eps$ in the envelope $u$,
and $y=x/\eps$ in $\chi_m$:
\begin{equation*}
\eps^{1-\mu/2-d/2}\int_{\R^d} |\xi|^\mu \left|
  u\left(t,z- \frac{\xi}{\sqrt\eps}\right)\right|^2 \left|
  \chi_m\left(y-\frac{\xi}{\eps}, p_0\right) \right|^2d\xi.
\end{equation*}
Setting $\zeta=\xi/\sqrt\eps$, this can be written as
\begin{equation*}
  \eps \int |\zeta|^\mu \left|
  u\left(t,z- \zeta\right)\right|^2 \left|
  \chi_m\left(y-\frac{\zeta}{\sqrt\eps}, p_0\right) \right|^2d\zeta.
\end{equation*}
Then, the following averaging result can be proved:
\begin{proposition}\label{prop:aver} 
Let \eqref{ass:simple} hold true and assume that $\zeta\mapsto 
|\zeta|^\mu \left|  u\left(t,z- \zeta\right)\right|^2 $ is in
$L^1(\R^d)$. Then, for all $k\in Y^*$, it holds 
\begin{equation*}
 \int_{\R^d} |\zeta|^\mu \left|
  u\left(t,z- \zeta\right)\right|^2 \left|
  \chi_m\left(y-\frac{\zeta}{\sqrt\eps}, k \right) \right|^2d\zeta \Tend
\eps 0 \int_{\R^d} |\zeta|^\mu \left|
  u\left(t,z- \zeta\right)\right|^2 d\zeta.
\end{equation*}
In addition, if $\zeta\mapsto |\zeta|^\mu \left|  u\left(t,z-
    \zeta\right)\right|^2 $ is in $W^{1,1}(\R^d)$,  
then the above convergence holds with an error of order $\O(\sqrt\eps)$.
\end{proposition}
This result can be seen as a variant of the \emph{two-scale convergence}
results introduced in \cite{Ng89,Al92}, and used in \cite{AlPi05}. The
main difference here is the convolution structure.
\begin{proof} We decompose $y\mapsto |\chi_m(y,k)|^2$ into its generalized 
Fourier series (recall that $\Gamma \simeq \mathbb Z^d$) and write
\begin{align*}
 \int |\zeta|^\mu \left|
  u\left(t,z- \zeta\right)\right|^2 \left|
  \chi_m\left(y-\frac{\zeta}{\sqrt\eps}\right) \right|^2d\zeta&=
\sum_{\gamma \in \Gamma}\int |\zeta|^\mu \left|
  u\left(t,z- \zeta\right)\right|^2
c_\gamma \e^{i \gamma \cdot (y-\zeta/\sqrt\eps)} d\zeta \\
&=
\sum_{\gamma \in \Gamma}c_\gamma \e^{i \gamma \cdot y}\int |\zeta|^\mu \left|
  u\left(t,z- \zeta\right)\right|^2
\e^{-i \gamma \cdot \zeta/ \sqrt\eps} d\zeta.
\end{align*}
By Riemann--Lebesgue lemma, for each term with $\gamma \not =0$, the
limit, as $\eps \to 0$, is 
zero. Then only the term corresponding to $\gamma=0$ remains, with
\begin{equation*}
  c_0(k) = \int_Y \left|
  \chi_m\left(y,k\right) \right|^2 dy = 1,
\end{equation*}
since the eigenfunctions $\chi_m(\cdot, k)$ form an orthonormal
basis of $L^2(Y)$. 
By the Dominated Convergence Theorem, we can exchange
the sum over $\gamma \in \Gamma$ and the limit $\eps\to 0$ in the
above computation provided that $\zeta\mapsto 
|\zeta|^\mu \left|  u\left(t,z- \zeta\right)\right|^2 $ is in $L^1$
(with an error $o(1)$). In the case where the function is in
$W^{1,1}$ we obtain an error $\O(\sqrt\eps)$. The reason is that
the coefficients $(c_\gamma)_{\gamma \in \Gamma}$
decrease rapidly for large $|\gamma|$, since $y\mapsto
|\chi_m(y,k)|^2$ is smooth, provided Assumption~\eqref{ass:simple}
holds true and  
thus we can perform an integration by parts, and use dominated
convergence again.  
\end{proof}
Assuming that $u$ is sufficiently smooth and decaying, we can use the
above averaging result and perform the same asymptotic expansion as
given in Section \ref{sec:multi} to arrive at  
the effective nonlinear Schr\"odinger equation
\begin{equation}\label{eq:NLprofile1}
  i\d_t u + \frac{1}{2}{\rm div}_z\((\nabla^2_k E_m\(p_0\))\cdot
  \nabla_z \) u = \lambda (|z|^\mu\ast |u|^{2}) u , \quad u_{\mid t=0} =  u_0.
\end{equation}
For $\mu<0$, existence of a smooth solution $u\in C([0,T_c),H^k)$,
locally in time, can be proved along the same lines as in
\cite{CazCourant} and hence, a result analogous to the one stated in
Proposition \ref{prop:app} is straightforward. 
In the case $\mu>0$ one can follow the arguments of \cite{Ma-p}, using a
functional framework which is more intricate, however (the Sobolev spaces
$H^k$ are not sufficient but have to be intersected with weighted
$L^2$ spaces), and we shall not do so here. In a similar spirit,
stability in the sense of Theorem~\ref{thm:stab} follows from an
adaptation of Lemma~\ref{lem:moser}, which we leave out.

\subsection{Smooth kernels} \label{sec:smoothK}
  If in (\ref{eq:r3}) the interaction kernel $K(x)$ is a given smooth
  function, bounded as well as its 
  derivative, then $\alpha_c=1$ (corresponding \emph{formally} to the
  case $\mu=0$). Such a situation appears for example in \cite{Ber99}, where
\begin{equation*}
  K(x) = \(a_1 + a_2 |x|^2 +a_3 |x|^4\) \e^{-A^2|x|^2}+a_4 \e^{-B^2|x|^2} , 
  \end{equation*}
  with constants $ a_1,a_2,a_3,a_4\in \R$, $A,B>0$.
Resuming the above computations in this context, we
  find:
  \begin{align*}
K\ast|\psi^\eps|^2 &= \eps^{-d/2}\int_{\R^d} K(\xi) \left|
  u\left(t,\frac{x-\xi-q(t)}{\sqrt\eps}\right)\right|^2 \left|
  \chi_m\left(\frac{x-\xi}{\eps}, p_0\right) \right|^2d\xi\\
& =\int_{\R^d} K(\zeta\sqrt\eps) \left|
  u\left(t,z-\zeta\right)\right|^2 \left|
  \chi_m\left(y-\frac{\zeta}{\sqrt\eps}, p_0\right) \right|^2d\zeta\\
&\Tend \eps 0 K(0)\int_{\R^d}  \left|
  u\left(t,z-\zeta\right)\right|^2 d\zeta= K(0)\|u(t)\|_{L^2}^2
=K(0)\|u_0\|_{L^2}^2 , 
  \end{align*}
due to mass conservation, along with an $\O(\sqrt\eps)$ convergence
rate under suitable 
assumptions. In particular, this shows that, as $\eps \to 0$, the
nonlinear effects become {\it negligible}.  
Indeed, in this case the envelope equation becomes 
\[
  i\d_t u + \frac{1}{2}{\rm div}_z\((\nabla^2_k E_m\(p_0\))\cdot
  \nabla_z \) u = K(0)\|u_0\|_{L^2}^2 u , \quad u_{\mid t=0} =  u_0.
\]
The right hand side involves a constant potential term, which can be
\emph{gauged away} via  
\[
v(t,x) = u(t,x) \e^{i  t K(0) \| u_0 \|^2_{L^2}}.
\]
The remaining amplitude $v(t,x)$ then solves a {\it free}
Schr\"odinger equation with effective mass tensor $\nabla^2_k
E_m\(p_0\)$.

\appendix

\section{Some useful algebraic identities}
\label{sec:appA}

For the derivation of the effective mass equation \eqref{eq:NLprofile}
we shall rely on several algebraic 
identities, which can be derived from Bloch's spectral problem (for
more details see, e.g., \cite{BeLiPa}):  
First, taking the gradient w.r.t. to $k$ of \eqref{eq:bloch}, we have
\begin{equation}\label{eq:dbloch}
  \nabla_k\(H(k) -E_m\) \chi_m + \(H(k) -E_m\) \nabla_k\chi_m=0
\end{equation}
and, by taking the in $L^2(Y)$-scalar product with $\chi_m$, we obtain
\begin{equation*}
  \nabla_k E_m = \<\chi_m, \nabla_k H(k) \chi_m\>_{L^2(Y)}
  +\<\chi_m,\(H(k) -E_m\) \nabla_k\chi_m\>_{L^2(Y)}.
\end{equation*}
Since $H(k) $
is self-adjoint, the last term is zero, thanks to \eqref{eq:bloch}. We
infer
\begin{equation}
  \label{eq:dkE}
  \nabla_k E_m(k) = \<\chi_m,\(-i \nabla_y+k\)\chi_m\>_{L^2(Y)}.
\end{equation}
Differentiating \eqref{eq:dbloch} again, we have, for all $j,\ell \in
\{1,\dots,d\}$: 
\begin{align*}
  \d_{k_j k_\ell}^2 \(H(k) -E_m\) \chi_m &+ \d_{k_j}\(H(k) -E_m\)
  \d_{k_\ell}\chi_m + \d_{k_\ell}\(H(k) -E_m\) \d_{k_j}\chi_m\\
&+\(H(k) -E_m\)
  \d_{k_j k_\ell}^2\chi_m=0. 
\end{align*}
Taking the scalar product with $\chi_m$, we have:
\begin{equation}
  \label{eq:d2kE}
  \begin{aligned}
   \d_{k_j k_\ell}^2 E_m(k) &= \delta_{j\ell} +\< \(-i
  \d_{y_j}+k_j\)\d_{k_\ell}\chi_m +\(-i
  \d_{y_\ell}+k_\ell\)\d_{k_j}\chi_m  ,\chi_m\>_{L^2(Y)}\\
&\quad -\< \d_{k_\ell}E_m\d_{k_j}\chi_m +\d_{k_j}E_m\d_{k_\ell}\chi_m
,\chi_m\>_{L^2(Y)}.
\end{aligned}
\end{equation}

\section{Adding an additional, slowly varying potential}
\label{sec:appB}
As a possible extension of our study, one might want to consider the
case where the wave function is not only under the  
influence of the nonlinearity and the periodic potential, but also add
an additional 
\emph{slowly varying} external potential $V(t,x)$, i.e.
\[
  i \eps\d_t \psi^\eps + \frac{1}{2} \Delta\psi^\eps = V_{\rm per}
     \left(\frac{x}{\eps}\right)\psi^\eps + V(t,x) \psi^\eps +
     \eps^\alpha f(|\psi^\eps|^2) \psi^\eps. 
\]
At least formally, this can be done by combining our analysis with the
results given in \cite{CaSp-p}: To this end, we define the
semi-classical band Hamiltonian 
\begin{equation*}
  h_m^{\rm sc}(k,x)= E_m(k) +V(t, x),\quad (k,x)\in Y^*\times\R^d,
\end{equation*}
and denote the corresponding semiclassical phase space trajectories by 
\begin{equation}\label{eq:traj}
\left\{  
\begin{aligned}
  \dot q(t)=\nabla_k E_m\(p(t)\),\quad q(0)=q_0,\cr
\dot p(t)= -\nabla_xV\(t,q(t)\),\quad p(0)=p_0.
\end{aligned}
\right.
 \end{equation} 
 This system is the analogue of the classical Hamiltonian phase space
 flow, in the 
 presence of an additional periodic potential $V_{\rm per}$. 
 In order to make sure that the system  \eqref{eq:traj} is
well-defined, it is sufficient to assume that $E_m(p(t))$ is a simple
eigenvalue ($|E_m(p(t))- E_n(k)| \not= 0$ for all $n\not=m$, $ t\in \R$,
$k\in Y^*$); see e.g. \cite{CaSp-p}, where examples of such situations are
given. 

The approximate solution under the form of a coherent state within the
$m$-th Bloch band is then given by: 
\[
  \varphi^\eps(t,x)=\eps^{-d/4} u 
\left(t,\frac{x-q(t)}{\sqrt\eps}\right) \chi_m\left(\frac{x}{\eps},
  p(t) \right) \e^{i \Phi_m(t,x) / \eps}
\]
with  $q(t), p(t)$ obtained from \eqref{eq:traj}. The highly
oscillatory phase takes the form 
$
\Phi_m(t,x) = S_m(t)+p(t)\cdot (x-q(t)),$
where $S_m(t)\in \R$ is the (purely time-dependent) semi-classical action
\[
S_m(t) = \int_0^t  p(s) \cdot \nabla E_m(p(s)) -
h_m^{\rm sc}\(p(s),q(s)\)  \, ds.
\]
Note that $\Phi_m$ simplifies to \eqref{eq:newphase} in the case where
$V(t,x)=0$. 
In this case, the governing equation for the leading profile $u(t,z)$
is found to be a nonlinear Schr\"odinger equation  
with time-dependent quadratic potential, time-dependent effective mass
$\nabla_k^2E_m(p(t))$ and coupling constant  
$\lambda_m(t)$, see \cite{CaSp-p} for more details. These features make it
difficult to give sufficient conditions under which where the solution
$u(t,z)$ is global, i.e. 
$T_c=\infty$. Indeed, the signature of $\nabla_k^2E_m(p(t))$ may
change, and the existence of Strichartz estimates for the linear part
is a non-trivial issue. Moreover, $\lambda_m(t)$ may 
also change sign, making the analysis even more delicate.

\end{document}